\documentclass[11pt]{article}
\usepackage{algorithm}
\usepackage{algpseudocode}
\usepackage{hyperref}
 \usepackage{amsmath}
 \usepackage{amsfonts}
 \usepackage{amssymb}
 \usepackage{amsbsy}
\usepackage{mathtools}
\usepackage{amsthm}
\newtheoremstyle{sfstyle} 
  {\topsep} 
  {\topsep} 
  {\itshape} 
  {} 
  {\sffamily\bfseries} 
  {.} 
  {1em} 
  {} 
\theoremstyle{sfstyle}
\newtheorem{theorem}{Theorem}
\newtheorem{lemma}{Lemma}
\theoremstyle{sfstyle}
\newtheorem{definition}{Definition}
\newtheorem{example}{Example}
\usepackage{pgfplots}
\usepackage{pgfplotstable}
\usepackage[draft]{changes}
\usepackage[normalem]{ulem}
\usepackage{enumitem}
\usepackage{multirow}
\usepackage{subcaption}
\usepackage{dsfont}
\usepackage{stackengine}
\usepackage{lipsum}
\DeclareMathAlphabet{\pazocal}{OMS}{zplm}{m}{n}

\usetikzlibrary{patterns}
\usetikzlibrary{calc}

\newcommand{\lt}{\left}
\newcommand{\rt}{\right}

\usepackage[margin=1in]{geometry}
\usepackage{natbib}
\mathtoolsset{showonlyrefs}
\usepackage{booktabs}
\usepackage{titlesec}
\titleformat{\section}{\sffamily\Large\bfseries}{\thesection.}{1em}{}
\titleformat{\subsection}{\sffamily\large\bfseries}{\thesubsection.}{1em}{}

\usepackage{caption}
\def\be{{\mathbf e}}

\def\bS{{\mathbf S}}

\def\bX{{\mathbf X}}
\def\bx{{\mathbf x}}

\def\by{{\mathbf y}}
\def\bZ{{\mathbf Z}}
\def\bz{{\mathbf z}}

\def\bLambda{{\boldsymbol \Lambda}}

\def\bmu{\boldsymbol \mu}

\def\cA{{\cal A}}
\def\cB{{\cal B}}
\def\cC{{\cal C}}

\def\cE{{\cal E}}
\def\cF{{\cal F}}

\def\cH{{\cal H}}

\def\cN{{\cal N}}

\def\cP{{\cal P}}

\def\cS{{\cal S}}

\def\cW{{\cal W}}

\def\sE{{\sf E}}

\def\sP{{\sf P}}

\def\bbR{\mathbb{R}}

\def\sone{\mathds{1}}
\def\bzero{\mathbf 0}

\def\rd{{\rm d}}
\def\Var{{\sf Var}}

\def\set{\cE}

\def\VR{\text{VR}}
\def\ER{\text{ER}}

\def\argmin{\arg\min}

\title{\sf\textbf{Wasserstein Distributionally Robust Rare-Event Simulation}}
\author{Dohyun Ahn\thanks{Corresponding author, E-mail: \href{mailto:dohyun.ahn@cuhk.edu.hk}{\tt dohyun.ahn@cuhk.edu.hk}},~~~Huiyi Chen\\
{\small \it The Chinese University of Hong Kong}
\smallskip\\
Lewen Zheng\\
{\small\it Huawei Hong Kong Research Center}}
\date{January 2026}

\begin{document}

\maketitle
\begin{abstract}
Standard rare-event simulation techniques require exact distributional specifications, which limits their effectiveness in the presence of distributional
uncertainty. To address this, we develop a novel framework for estimating rare-event probabilities subject to such distributional model risk. Specifically,
we focus on computing worst-case rare-event probabilities, defined as a distributionally robust bound against a Wasserstein ambiguity set centered at a specific nominal
distribution. By exploiting a dual characterization of this bound, we propose Distributionally Robust Importance Sampling (DRIS), a computationally tractable methodology designed to substantially reduce the variance associated with estimating the dual components. The proposed method is simple to implement and requires low sampling costs. Most importantly, it achieves \emph{vanishing relative error}---the strongest efficiency guarantee that is notoriously difficult to establish  in rare-event simulation. Our numerical studies confirm the superior performance of DRIS against existing
benchmarks.
\end{abstract}

\section{Introduction}
From managing financial tail risk to predicting extreme climate events, quantifying the likelihood of rare events is critical for system stability and safety~\citep{Glasserman2003-MCFE,Asmussen2007,Rubino2009}.
The fundamental mathematical task involves estimating the probability that a random vector falls into a critical rare-event set. Since standard Monte Carlo methods are computationally inefficient for such tasks, sophisticated variance reduction techniques---such
as importance sampling, conditional Monte Carlo, splitting, and stratification---have been developed for various models and problems; see, e.g., \cite{Glasserman2000-var-red,Glasserman2008,Juneja:02,bassamboo_portfolio_2008,Blanchet2014-queue,bai_rare-event_2022,Ahn2023,deo2025achieving} and references therein. 

However, a significant theoretical gap persists: these classical methods assume precise knowledge of the underlying probability distributions, making them vulnerable to model misspecification. In real-world scenarios, such granular information is rarely available---particularly when data are scarce or noisy---resulting in distributional uncertainty. To overcome this limitation, we employ \emph{a distributionally robust approach to rare-event simulation}. To be more specific, we focus on efficiently computing worst-case rare-event probabilities over a family of plausible distributions, mathematically formalized as a Wasserstein ball surrounding a nominal distributional model. To the best of our knowledge, this is the first study to introduce an efficient Monte Carlo approach for rare-event probability estimation in the presence of distributional model risk.

In terms of developing simulation methods for  worst-case expectations under model uncertainty, our approach is closely related to those of \cite{Glasserman2014} and \cite{Blanchet2017}. The former proposes the so-called robust Monte Carlo to estimate risk measures over distributional ambiguity sets defined by relative entropy and $\alpha$-divergence, while the latter focuses on computing worst-case expectations of two random vectors with fixed marginals but unknown dependence structures. Despite such methodological developments, neither of these prior studies specifically target variance reduction for rare-event simulation; consequently, their efficacy in this regime remains unestablished.

Regarding distributional robustness specifically for rare-events, existing literature has predominantly relied on optimization-based or extreme-value-theory-based approaches rather than simulation methodologies; see, for instance, 
\cite{Lam2017-robust-tail,BlanchetHM2020} and \cite{Bai2023}. 
Concurrently, a recent study by 
\cite{Huang2023} 
utilizes random walk tail probabilities to analyze the vulnerability of rare-event probabilities to tail uncertainty, arguing that heavy-tailed cases exhibit a higher sensitivity to model misspecification than light-tailed cases. In contrast, we put an emphasis on simulation and bridge the gap by proposing a  variance reduction technique for estimating worst-case rare-event probabilities.

Specifically, this paper develops a novel importance sampling method, which we call \emph{Distributionally Robust Importance Sampling (DRIS)}, to estimate the aforementioned worst-case rare-event probabilities for convex target sets. Leveraging a general duality result for Wasserstein distributionally robust optimization, the probability of interest can be reformulated as the probability of a neighborhood of the target set under the nominal distribution. From a computational viewpoint, this dual reformulation requires a two-step process: first estimating the neighborhood and then incorporating it into the final probability computation. Since both steps involve rare-event simulation, our DRIS method is designed to address these requirements via a cohesive, computationally efficient, and easy-to-implement algorithm. 

Most importantly, we establish that the DRIS estimator admits a central limit theorem and exhibits \emph{vanishing relative error} (Theorems~\ref{thm:clt} and~\ref{thm:bre}). These main theoretical results are built upon (i) empirical process theory with Vapnik–Chervonenkis-type arguments and (ii) rare-event analysis in simulation. It is worth emphasizing that the property of vanishing relative error, which ensures the relative error decays to zero as the target event becomes increasingly rare, is arguably the highest notion of efficiency in rare-event simulation and is seldom achieved in prior studies.

The remainder of the paper is organized as follows. Section~\ref{sec:problem} formulates the main problem. In Section~\ref{sec:prelim}, we review strong duality for Wasserstein distributionally robust optimization in the context of worst-case probabilities and present preliminary theoretical results. Section~\ref{sec:main analysis} introduces the proposed DRIS procedure and establishes its theoretical performance guarantees in the rare-event regime. In Section \ref{sec:numerical}, we numerically validate the effectiveness of the algorithm. Section \ref{sec:conclude} concludes the paper. All proofs are deferred to the appendices.

\section{Problem Formulation}\label{sec:problem}
Let $\cP$ denote the set of all probability distributions supported on the $n$-dimensional Euclidean space. Then, the 2-Wasserstein distance between $\sP_0,\sP\in\cP$ is defined as 
\begin{equation}	
\cW_2(\sP_0,\sP)=\inf_{\pi\in\Pi(\sP_0,\sP)}\lt(\sE_{(\bX_0,\bX)\sim\pi}\big[\|\bX_0-\bX\|^2\big]\rt)^{1/2},
\end{equation}
where $\Pi(\sP_0,\sP)$ is the set of all couplings of $\sP_0$ and $\sP$, that is, the set of all joint distributions with marginals $\sP_0$ and $\sP$, respectively. Accordingly, the 2-Wasserstein ball of radius $\delta>0$ centered at the nominal distribution $\sP_0$ is given by
$$
\cB_\delta(\sP_0) = \{\sP\in\cP:\cW_2(\sP_0,\sP)\leq\delta\}.
$$

In this paper, we investigate the estimation of the worst-case probability defined by:
\begin{equation}\label{eq:basic-quantity}
	p_*=\sup_{\sP\in\cB_\delta(\sP_0)}\sP(\bX\in\set),
\end{equation}
where $\delta\in(0,\infty)$ is a fixed constant, $\set$ is a nonempty, full-dimensional, closed, and convex set that does not contain the origin, and $\sP_0$ is the $n$-dimensional standard normal distribution. This quantity corresponds to a version of the inner worst-case problem
in Wasserstein distributionally robust optimization, which has received considerable attention in recent literature \citep{zhangYG2025}. 
Although we focus on Gaussian nominal distributions, the proposed methodology extends naturally to other multivariate elliptical families. We prioritize the Gaussian setting due to its prevalence in the OR/MS literature, where critical metrics often correspond to rare-event probabilities governed by standard normal distributions~\citep[][Chapter 9]{Bucklew2004}. Below is one of such examples in finance:
\begin{example}\label{ex:portfolio}
According to \cite{Glasserman2000-var-red}, the loss of a portfolio of European call/put options over the time interval $[t,t+\rd t]$ can be approximated by
$$
L\coloneqq V(\bS_t,t) - V(\bS_t+\rd\bS,t+\rd t) \approx -\frac{\partial V}{\partial t} \rd t -\Delta^\top \rd\bS-\frac12\rd\bS^\top\Gamma \rd\bS \eqqcolon \tilde L,
$$
where $\bS_t$ and $V(\bS_t,t)$ denote the values of $n$ risk factors and the portfolio value, respectively, $\rd\bS = \bS_{t+\rd t}-\bS_t$, $\Delta = \nabla_{\bS}V^\top$, $\Gamma = \nabla^2_{\bS}V$, and ``$\approx$'' holds by the delta-gamma approximation.
If $\rd\bS = \bS_{t+\rd t}-\bS_t$ follows a multivariate normal and the approximation is exact (i.e., $L=\tilde L$), 
    $$
    \sP(L>\ell) = \sP\lt(a + \sum_{i=1}^n\lt(b_iX_i + c_iX_i^2\rt)>\ell\rt),
    $$
for a loss threshold $\ell>0$, fixed constants $a,b_1,\ldots,b_n,c_1,\ldots,c_n$ with $c_1,\ldots,c_n\leq0$, and $X_1,\ldots,X_n\stackrel{{\sf iid}}{\sim}\cN(0,1)$.
This quantity is commonly used to define a portfolio risk measure, and when $\ell$ is large, it becomes a probability that independent standard normals belong to a convex rare-event set.
\end{example}

In addition to this example, many continuous-time stochastic models, such as geometric Brownian motion and Gaussian Markov processes, can be simulated as weighted sums of standard normal variables via the Euler scheme, which is essential not only for financial modeling but also for analyzing system stability in other domains: heavy-traffic approximations in queueing theory rely on diffusion processes driven by Brownian motion, and demand processes in supply chain management are often modeled as Gaussian random walks.

It is worth noting that if $\bX$ follows an $n$-dimensional non-standard normal distribution, one can find $\bmu\in\bbR^n$ and $\bLambda\in\bbR^{n\times m}$ with $n\geq m$ such that $\bX$ has that same distribution as $\bmu+\bLambda\widetilde\bX$, where $\widetilde\bX$ follows an $m$-dimensional standard normal distribution. Accordingly, the probability $\sP(\bX\in\set)$ coincides with the probability that $\widetilde\bX$ belongs to another convex set given by $\{\bx:\bmu+\bLambda\bx\in\set\}$. Consequently, restricting the analysis to the standard normal distribution suffices for all Gaussian models.

Without loss of any generality, we assume that $\bx^*\coloneqq\argmin_{\bx\in\set}\|\bx\|$ lies on the $x_1$-axis. It can be satisfied through a suitable rotation of the coordinates and a rearrangement of the variables, which does not affect \eqref{eq:basic-quantity} because the standard normal distribution is invariant under such transformations. Furthermore, we focus on a situation where $\{\bX\in\set\}$ is a rare event in the sense that its likelihood is close to zero. We study this mathematically by
considering a sequence of sets indexed by a rarity parameter $r > 0$:
\begin{equation}\label{eq:Er}
\set_r = \lt\{\frac{r}{\|\bx^*\|}\bx:\bx\in\set\rt\},
\end{equation}
in which case $(r,0,\ldots,0)=\argmin_{\bx\in\set_r}\|\bx\|$. Hence, the set $\set_r$ moves away from the origin as $r\to\infty$, leading to $\lim_{r\to\infty}\sP_0(\bX\in\set_r)=0$. 

To analyze the efficiency of the proposed estimator, we adopt the following performance criterion widely used in the rare-event simulation literature~\citep[see, e.g.,][]{bassamboo_portfolio_2008,NakayamaT2023}:
\begin{definition}\label{def:efficiency}
Let $q_r$ denote a quantity of interest satisfying $q_r\to0$ as $r\to\infty$. Suppose that an unbiased estimator $Q_{N,r}$ for $q_r$, constructed by $N$ iid samples, admits a central limit theorem with asymptotic variance $\xi_r^2$ for any $r>0$; that is,
$\sqrt{N}(Q_{N,r}-q_r)\Rightarrow\cN(0,\xi_r^2)$ as $N\to\infty$, where $\Rightarrow$ represents convergence in distribution, and $\cN(\gamma,\nu^2)$ means a normal random variable with mean $\gamma$ and variance $\nu^2$.
	Then, we say that $Q_{N,r}$ has \textit{vanishing relative error} if
	$$\limsup_{r\to\infty}\frac{\xi_r}{q_r}=0.$$
\end{definition}
Vanishing relative error is often regarded as the highest efficiency notion in the context of rare-event simulation. As noted in~\cite{botev_normal_2017}, Monte Carlo estimators for light-tailed distributions seldom exhibit vanishing relative error. This property ensures that, given a fixed large sample size, the accuracy of the associated estimator improves as the target event becomes rarer.

\section{Preliminaries}\label{sec:prelim}
In this section, we review a strong duality result for our target quantity in \eqref{eq:basic-quantity} and introduce our preliminary theoretical analysis. Both play a crucial role in making the problem tractable and facilitating the main analysis in Section~\ref{sec:main analysis}. Before delving into the details, let us briefly introduce our notational conventions used throughout the paper. We denote by $\sE_0$ the expectation under the nominal distribution $\sP_0$, and we use $d(\bx,\cS)=\min_{\by\in\cS}\|\bx-\by\|$ to represent the distance between a point $\bx\in\bbR^n$ and a set $\cS\subset\bbR^n$. Also, for brevity, we write $\sE_0[g(\bX);\cA] \coloneqq \sE_0[g(\bX)\sone\{\cA\}]$ for any function $g$ and any event $\cA$.

{\sf\textbf{Strong duality for \eqref{eq:basic-quantity}.}}
The optimization problem in \eqref{eq:basic-quantity} is infinite-dimensional and thus intractable to solve directly. Fortunately, established results in the literature on Wasserstein distributionally robust optimization demonstrate that the dual formulation of \eqref{eq:basic-quantity} is computationally tractable. We restate a version of these results in our framework and discuss its implications for rare-event simulation.

\begin{lemma}[Theorem 2 of \cite{Blanchet2019-DRO}]\label{lem:dual}
	Let $h(u)=\sE_0[d(\bX,\set)^2;d(\bX,\set)\leq u]$ and $p(u)=\sP_0(d(\bX,\set)\leq u)$. Then, the probability $p_*$
 in \eqref{eq:basic-quantity} is equal to $p({u_*})$, where ${u_*}=h^{-1}(\delta^2)$.
\end{lemma}

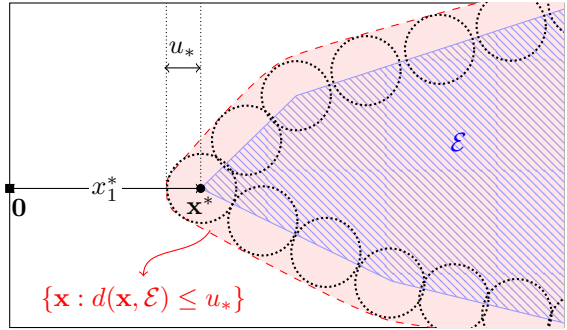
\begin{figure}
    \caption{A graphical illustration of the relationship between the target set and its inflated version based on the duality result\label{fig:dual}}
    \centering
    \begin{tikzpicture}[font=\small]
		\tikzstyle{ann} = [fill=white,font=\small,inner sep=1pt]
			\begin{axis}[name=plot1,width=2*1.45in, height=2*0.85in,
				scale only axis, ticks=none,
				xmin=0, xmax=5.8,
				ymin=0.5, ymax=4,
				axis background/.style={fill=white}]
				\addplot+ [red,mark=none,dashed,fill=red!10] plot [smooth, tension=0.3] coordinates {
					(6,4.3) (2.9,3.4) (1.72,2.2)
                    (1.65,2)(1.79,1.7) (4,0.57) (7,0.5) (6,4.2)};
				\addplot+ [draw=blue!40, mark=none,pattern=north west lines, pattern color=blue!40] coordinates {
					(6,4) (3,3) (2,2) (4,1) (6,0.5) (6,4)};
				\draw[black,densely dotted,thick] (axis cs:6,4) circle[radius=0.46cm];
				\draw[black,densely dotted,thick] (axis cs:5.25,3.75) circle[radius=0.46cm];
				\draw[black,densely dotted,thick] (axis cs:4.5,3.5) circle[radius=0.46cm];
				\draw[black,densely dotted,thick] (axis cs:3.73,3.27) circle[radius=0.46cm];
				
				\draw[black,densely dotted,thick] (axis cs:3,3) circle[radius=0.46cm];
				\draw[black,densely dotted,thick] (axis cs:2.48,2.52) circle[radius=0.46cm];
				
				\draw[black,densely dotted,thick] (axis cs:2.01,2) circle[radius=0.46cm];
				\draw[black,densely dotted,thick] (axis cs:2.65,1.65) circle[radius=0.46cm];
				\draw[black,densely dotted,thick] (axis cs:3.31,1.31) circle[radius=0.46cm];
				
				\draw[black,densely dotted,thick] (axis cs:4-0.02,1-0.02) circle[radius=0.46cm];
				\draw[black,densely dotted,thick] (axis cs:4.65,0.81) circle[radius=0.46cm];
				\draw[black,densely dotted,thick] (axis cs:5.31,0.65) circle[radius=0.46cm];
				\draw[black,densely dotted,thick] (axis cs:6,0.5) circle[radius=0.46cm];
			\draw[arrows=<->](axis cs:2-0.38,3.3)--(axis cs:2,3.3);
			\draw[arrows=<->](axis cs:0,2)--(axis cs:2,2);
			\node at (axis cs:1.81,3.5) {$u_*$};
			\node[ann] at (axis cs:1,2) {$x_1^*$};
			\draw[black,densely dotted] (axis cs:2,2)--(axis cs:2,4);
			\draw[black,densely dotted] (axis cs:2-0.36,2)--(axis cs:2-0.36,4);
            \addplot+[black, mark=square*, mark size=1.5, every mark/.append style={solid, fill=black}] coordinates {(0,2)};
            \addplot+[black, mark=*, mark size=1.5, every mark/.append style={solid, fill=black}] coordinates {(2,2)};
			\node at (axis cs:0.1,1.8) {{\small${\bf0}$}};
			\node at (axis cs:2,1.82) {{\small$\bx^*$}};
				\node at (axis cs:4.7,2.5){$\color{blue}{\set}$};
				\node (A) at (axis cs:1.4,0.8){$\color{red}{\{\bx:d(\bx,\set)\leq u_*\}}$};
            \draw[->,red] (axis cs:2.1,1.55) to [out=240,in=90] (A.north);
			\end{axis}
		\end{tikzpicture}
\end{figure}

The significance of this duality result lies in expressing the worst-case probability $p_*$ as the probability, under the nominal distribution $\sP_0$, of an inflated superset of the target event, given by $\{\bx:d(\bx,\set)\leq{u_*}\}$. Figure \ref{fig:dual} illustrates the connection between the target set and its inflated counterpart: the blue slashed region depicts the target set $\set$, while the red shaded area corresponds to its inflated version $\{\bx:d(\bx,\set)\leq{u_*}\}$. The dotted circles represent a radius of $u_*$; the union of such circles centered at all points in $\set$ characterizes the inflated superset. Based on the assumption in Section~\ref{sec:problem}, $\bx^*$ lies on the $x_1$-axis, and hence, its distance from the origin is $x_1^*$.

Since Lemma~\ref{lem:dual} holds for any set $\set$, the function $h(\cdot)$ and the value $u_*$ in the lemma are similarly defined for the sequence of sets $\{\set_r\}_{r>0}$ in \eqref{eq:Er} as follows: for $r>0$ and $\delta,u\geq0$, we let
$$h_r(u)=\sE_0[d(\bX,\set_r)^2;d(\bX,\set_r)\leq u]
~~\text{and}~~u_r = h_r^{-1}(\delta^2).
$$
Then, by the above lemma, we have
\begin{equation}\label{eq:def-pr}
p_r\coloneqq \sup_{\sP\in\cB_\delta(\sP_0)}\sP(\bX\in\set_r)=\sP_0(d(\bX,\set_r)\leq u_r).
\end{equation}
Although $u_r$ and $p_r$ depend on the radius $\delta$ of the 2-Wasserstein ball, this dependence is suppressed in the notation.

{\sf\textbf{Preliminary theoretical results.}} Given our rare-event regime where $r$ tends to $\infty$, we analyze how $u_r$ and $p_r$ behave as $r$ grows. Firstly, the following lemma describes the behavior of $u_r$:
\begin{lemma}[Asymptotic Behavior of $u_r$]\label{lem:u}
	For any $\delta,M>0$, there exists $r_0>0$ such that for all $r\geq r_0$,
	\begin{equation}
		\label{eq:u-bound}
		M<r-u_r<\bar\Phi^{-1}\lt(\frac{\delta^2}{r^2}\rt),
	\end{equation}
where $\bar\Phi(\cdot)$ denotes the standard normal complementary cumulative distribution function.
\end{lemma}
Observe that $r-u_r$ represents the distance between the origin and the inflated version of $\set_r$. Hence, by the first inequality in \eqref{eq:u-bound}, Lemma~\ref{lem:u} confirms that the inflated superset moves away from the origin as $r$ increases, which suggests that \emph{$p_r$ in \eqref{eq:def-pr} is again a rare-event probability}. This motivates us to develop an efficient rare-event simulation algorithm for estimating this probability.

Furthermore,  as shown in Appendix \ref{app:main}, $\bar\Phi^{-1}(\delta^2/r^2)$ in \eqref{eq:u-bound} grows sublinearly as $r\to\infty$. Consequently, the second inequality in \eqref{eq:u-bound} implies that this distance diverges at a sublinear rate. This indicates that the worst-case probability $p_r$ decays slower than the exponential rate of the nominal probability $\sP_0(\bX\in\set_r)$. We formalize this observation in the following theorem, which characterizes the asymptotic lower bound for $p_r$ as $r\to\infty$.
\begin{theorem}[Asymptotic Behavior of $p_r$]\label{thm:prob-asymp}
For any $\delta>0$,	$\liminf_{r\to\infty}r^2 p_r\geq \delta^2$.
\end{theorem}
According to this theorem, achieving vanishing relative error (Definition~\ref{def:efficiency}) for the estimation of $p_r$ requires the construction of an $N$-sample-based unbiased estimator whose asymptotic variance decays at a rate faster than $r^{-4}$. In the next section, we propose a novel importance sampling estimator that satisfies this condition.

\section{Main Algorithm and Results}\label{sec:main analysis}
Lemma~\ref{lem:dual} allows us to compute the worst-case probability $p_*$ in two steps: (a) solving $h(u)=\delta^2$ to obtain ${u_*}$ and (b) evaluating $p({u_*})$. Both tasks involve the estimation of expectations under the nominal distribution $\sP_0$ defined over the rare-event sets of the form $\{\bx:d(\bx,\set)\leq u\}$ (see Section~\ref{sec:prelim}). Accordingly, in this section, we propose a comprehensive and tractable algorithm that addresses these two rare-event estimation steps and demonstrate that it achieves vanishing relative error.

\subsection{DRIS Algorithm}\label{subsec:DRIS}
For the above-mentioned tasks, sampling $\bX$ in the vicinity of the rare-event set $\{\bx:d(\bx,\set)\leq u\}$ is essential for any feasible $u$. We identify $X_1$ as the primary driver of the said rare event since $\{\bx:d(\bx,\cE)\leq u\}\subseteq\{\bx:x_1\geq x_1^*-u\}$ holds for all $u$. 
Moreover, the rare-event set $\{\bx:d(\bx,\set)\leq u\}$ is the Minkowski sum of two convex sets $\set$ and $\{\bx:\|\bx\|\leq u\}$, and therefore, is also convex.
Consequently, inspired by the conditional importance sampling method in \cite{ahnZwsc:23}, our importance sampling approach involves: (a) generating $X_1$ via $X_1 = x_1^* - u + Y/(x_1^* - u)$, with $Y$ drawn from the standard exponential distribution; and (b) sampling $(X_2,\ldots,X_n)$ from the standard normal distribution. 

We then define $\bZ=(Y,X_2,\cdots,X_n)^\top$ and denote the expectation with respect to its distribution by $\sE$. We also define a transformation $f_u:\bbR^n\rightarrow\bbR^n$ as 
\begin{equation}\label{eq:f_u}    
f_u(\bz)=\lt(x_1^*-u+\frac{z_1}{x_1^*-u},z_2,\cdots,z_d\rt)^\top,
\end{equation}
which maps $\bZ$ to $\bX$. Finally, let
\begin{equation}\label{eq:L_u}
L_u(\bz)\coloneqq\frac{\exp(-z_1^2/(2(x_1^*-u)^2)-(x_1^*-u)^2/2)}{(x_1^*-u)\sqrt{2\pi}}\sone\{z_1\geq 0\}    
\end{equation}
be the likelihood ratio associated with our importance sampling approach.
In this setup, it is easy to see that
\begin{equation}
 \lt\{~\begin{aligned} h(u)&=\sE[d(f_u(\bZ),\set)^2\sone\{d(f_u(\bZ),\set)\le u\}L_u(\bZ)];\\
 p(u)&=\sE[\sone\{d(f_u(\bZ),\set)\le u\}L_u(\bZ)].
 \end{aligned}\rt.
\end{equation}

This forms unbiased estimators for $h(u)$ and $p(u)$ and enables us to develop the following estimation procedure for $p_*$:
\begin{enumerate}[label=(\roman*)]
    \item Take $N$ iid copies of $\bZ$, denoted by $\{\bZ_i\}_{i=1}^n$;
    \item Let $H(\cdot,u) \coloneqq d(f_u(\cdot),\set)^2\sone\{d(f_u(\cdot),\set)\leq u\}L_u(\cdot)$ for $u\geq0$ and define an estimate of $h(\cdot)$ as
\begin{equation}\label{eq:h_hat_N}
	\widehat h_N(u)=\frac{1}{N}\sum_{i=1}^{N} H(\bZ_i,u)~~\text{for}~u\geq0;
\end{equation}
    \item Compute the estimate $\widehat u_N \coloneqq \inf\{u:\widehat h_N(u)>\delta^2\}$ for $u_*$;
    \item Let $P(\cdot,u)\coloneqq\sone\{d(f_u(\cdot),\set)\leq u\}L_u(\cdot)$ and define an estimate of $p(u)$ as
\begin{equation}\label{eq:p_hat_N}
	\widehat p_N(u)=\frac{1}{N}\sum_{i=1}^{N} P(\bZ_i,u)~~\text{for}~u\geq0;
\end{equation}
\item Calculate the estimate of the worst-case probability $p_*$ by evaluating $\widehat p_N(\widehat u_N)$.
\end{enumerate}
We refer to this method and the estimator $\widehat p_N(\widehat u_N)$ as \emph{Distributionally Robust Importance Sampling (DRIS)} and the DRIS estimator, respectively. We detail its procedure in Algorithm \ref{alg:CIS-normal}. It is important to highlight that while Step (iii) involves root-finding, it requires no additional sampling costs, in contrast to typical root-finding procedures coupled with importance sampling \citep{HeJLF2024}. Moreover, the implementation of the DRIS method is computationally cheap: although it involves a root-finding procedure, the algorithm avoids costly operations elsewhere. Particularly, our sampling distributions (i.e., exponential and normal distributions) are straightforward to simulate, ensuring low sampling costs.

\begin{algorithm}[t]
\caption{Distributionally Robust Importance Sampling (DRIS)\label{alg:CIS-normal}}	
	\begin{algorithmic}[1]
		\State \textbf{Input:} $N$, $x_1^*$, and $\delta$
		\State Generate $N$ samples $\{y_i\}_{i=1}^N$ of $Y$  from the standard exponential distribution
		\State Take $N$ samples $\{(x_{2,i},\ldots,x_{n,i})\}_{i=1}^N$ of $(X_2,\ldots,X_n)$ from the $(n-1)$-dimensional standard normal distribution
        \State Set $\bz_i=(z_{1,i},\ldots,z_{n,i})$ for $i=1,\ldots,N$, where $z_{1,i}=y_i$ and $z_{j,i}=x_{j,i}$ for $j=2,\ldots,n$
        \State Set $h_N(u) = N^{-1}\sum_{i=1}^Nd(f_u(\bz_i),\set)^2\sone\{d(f_u(\bz_i),\set)\leq u\}L_u(\bz_i)$ for any $u\geq0$, where $f_u(\cdot)$ and $L_u(\cdot)$ are defined as in \eqref{eq:f_u} and \eqref{eq:L_u}, respectively
        \State Find $u_N=\inf\{u:h_N(u)>\delta^2\}$ via a (deterministic) root-finding procedure
		\State\textbf{Return:} $p_N = N^{-1}\sum_{i=1}^N\sone\{d(f_{u_N}(\bz_i),\set)\leq u_N\}L_{u_N}(\bz_i)$
	\end{algorithmic}
\end{algorithm}

\subsection{Efficiency of DRIS}
We now show that our proposed methodology has strong theoretical performance guarantees, satisfying the efficiency criterion in Definition~\ref{def:efficiency}. To that end, we first characterize the central limit theorem for the DRIS estimator $\widehat p_N(\widehat u_N)$ in the following result:

\begin{theorem}[Central Limit Theorem]\label{thm:clt}
	Suppose that there exist $u_L,u_U\in (0,x_1^*)$ such that ${u_*}, \widehat u_N\in[u_L,u_U]$ for all sufficiently large $N$. Then, 
	\begin{equation}
		\label{eq:CLT-var}
		\sqrt{N}(\widehat p_N(\widehat u_N)-p_*) \Rightarrow \cN\lt(0,\Var\lt(P(\bZ,{u_*})-\frac{H(\bZ,{u_*})}{u_*^2}\rt)\rt)~~\text{as}~N\to\infty.
	\end{equation}
\end{theorem}

It is straightforward to verify that the central limit theorem stated above holds in our asymptotic regime with the sequence of sets $\{\set_r\}_{r>0}$. Specifically, for all $r>0$, the DRIS estimator for $p_r$ in \eqref{eq:def-pr} has asymptotic variance
$$
\sigma_r^2\coloneqq\Var\lt(\sone\{d(f_{u_r}(\bZ),\set_r)\leq u_r\}L_{u_r}(\bZ)\lt(1-\frac{d(f_{u_r}(\bZ),\set_r)^2}{u_r^2}\rt)\rt).
$$
Based on this asymptotic variance, the following theorem presents the main finding of this paper: a characterization of the asymptotic efficiency of our DRIS estimator. This result demonstrates the effectiveness of using a fixed set of samples for estimating both $h(\cdot)$ and $p(\cdot)$.
\begin{theorem}[Vanishing Relative Error]\label{thm:bre}
	For any $\delta>0$, $\limsup_{r\to\infty}r^2(r-u_r)^2\sigma_r^2/p_r^2<\infty$.
\end{theorem}
Since $r-u_r\to\infty$ as $r\to\infty$ (Lemma~\ref{lem:u}), the preceding theorem shows that the relative error of the DRIS estimator asymptotically changes at a rate at most $r^{-1}(r-u_r)^{-1}$ as $r\to\infty$, implying that the DRIS estimator achieves vanishing relative error.

\section{Numerical Experiments}\label{sec:numerical}
In this section, we conduct numerical experiments to validate the performance of the proposed method. To numerically compare the DRIS method with the application of existing Monte Carlo methods, we report two performance indicators for each experiment conducted below: variance ratio (VR) and efficiency ratio (ER). For a crude Monte Carlo estimator $Z^{\tt MC}$ with runtime $\tau^{\tt MC}$ and a target estimator $Z$ with runtime $\tau$, we define $\text{VR}\coloneqq\Var(Z^{\tt MC})/\Var(Z)$ and $\text{ER}\coloneqq\text{VR}\times \tau^{\tt MC}/\tau$. We also report the relative error of an estimator $Z$ at the 95\% confidence level defined as $1.96\sqrt{\Var(Z)}/\sE[Z]$. While ER is often considered a more comprehensive measure of efficiency, computation time is sensitive to hardware performance and implementation details; therefore, we present VR as a critical complementary metric.

\subsection{Experimental Setups}\label{subsec:experiment setup}
We use the following two examples for our numerical experiments.

{\sf\textbf{A toy example.}} We first consider a simple two-dimensional setup where the target set is given by
$\set_r = \{\bx\in\bbR^2:x_1-5x_2\geq r, x_1+5x_2\geq r\}$ 
and the radius of the 2-Wasserstein ball is set as $\delta = 0.001$.
We obtain the estimates of $\widehat u_N$ and $\widehat p_N(\widehat u_N)$ using the sample size of $10^7$ and replicate the entire procedure for $100$ times to calculate the average runtime and variance for each algorithm.  To the
best of our knowledge, there are no particular simulation methods developed to estimate Wasserstein distributionally robust rare-event probabilities. Hence, we compare the performance of the DRIS method with those of crude Monte Carlo (MC) and classical exponential twisting (ET) schemes, both of which are applied to estimate $h(\cdot)$ and $p(\cdot)$ analogously to the DRIS method in \eqref{eq:h_hat_N} and \eqref{eq:p_hat_N}.

{\sf\textbf{Portfolio loss probabilities.}} We next revisit Example~\ref{ex:portfolio} in Section~\ref{sec:problem} to estimate portfolio loss probabilities. We construct a portfolio consisting of $n=5$ uncorrelated underlying assets, adopting the parameter settings from \cite{Glasserman2000-var-red}. Specifically, we assume 250 trading days per year, a risk-free rate of 5\%, and $\rd t=0.04$. Each underlying asset has an initial value $S_0=100$ and volatility $\sigma=0.3$. For each asset, the portfolio holds long positions in 10 at-the-money call options and 5 at-the-money put options. All options have a half-year maturity. The loss threshold $\ell$ is set to 120 in all cases. To align with our rare-event setting, we scale the risk factor $\bX$ by $r^{-1}$ for various values of $r$.
Finally, we set $\delta = 0.01$ and use the same benchmarks, sample size, and number of macroreplications as in the previous toy example.

\subsection{Summary of the Numerical Results}

Tables~\ref{tab:polyhedron} and~\ref{tab:portfolio} report the estimates of $u_r$ and $p_r$ and the runtimes of the algorithms, along with the corresponding 95\% relative error, VR and ER, for the two examples described in Section~\ref{subsec:experiment setup}. 
In all cases we consider, our proposed method completely dominates the two benchmarks, demonstrating greater variance reduction and higher efficiency. This significant performance gap between DRIS and the other two methods, which widens as $r$ increases, validates our theoretical results. Although ET performs competitively in our numerical experiments, its performance in these problems lacks theoretical justification, and more importantly, DRIS consistently yields superior results. The increased runtimes for ET and DRIS, compared to MC, arise because the root-finding procedure embedded in these algorithms requires transforming samples and solving the distance dependent on the evaluated $u$; in contrast, samples in the crude Monte Carlo algorithm remain unchanged.

\begin{table}[htbp]
	\centering\small
	\caption{Numerical results for the toy example in Section~\ref{subsec:experiment setup}}
    \label{tab:polyhedron}
	\begin{tabular}{ccccccc}\toprule
		Method&$r$&$u_r$ (95\% rel.~err.)&$p_r$ (95\% rel.~err.)&Time (sec)& $\VR$ &$\ER$\\ \midrule
		\multirow{4}*{MC}
		& $2$ & $0.0027$ ($1.62\%$) & $2.40\times 10^{-3}$ ($1.13\%$) & $11$ & -- & -- \\
 	    & $3$ & $0.0141$ ($3.29\%$) & $2.39\times 10^{-4}$ ($2.97\%$) & $12$ & -- & -- \\
 	    & $4$ & $0.0931$ ($8.11\%$) & $2.36\times 10^{-5}$ ($7.80\%$) & $12$ & -- & -- \\
 	    & $5$ & $0.5245$ ($12.43\%$) & $3.10\times 10^{-6}$ ($19.91\%$) & 13 & -- & -- \\
		\midrule
		\multirow{4}*{ET}
		& $2$ & $0.0027$ ($0.42\%$) & $2.40\times 10^{-3}$ ($0.29\%$) & $144$ & $16$ & $1.3$ \\
 	    & $3$ & $0.0147$ ($0.30\%$) & $2.40\times 10^{-4}$ ($0.27\%$) & $151$ & $125$ & $9.8$ \\
 	    & $4$ & $0.0965$ ($0.18\%$) & $2.31\times 10^{-5}$ ($0.15\%$) & $144$ & $2,600$ & $225$ \\
 	    & $5$ & $0.5163$ ($0.08\%$) & $3.08\times 10^{-6}$ ($0.07\%$) & $112$ & $78,036$ & $9,107$ \\
		\midrule
		\multirow{4}*{DRIS}
		& $2$ & $0.0027$ ($0.24\%$) & $2.41\times 10^{-3}$ ($0.16\%$) & $149$ & $48$ & $3.7$ \\
 	    & $3$ & $0.0146$ ($0.15\%$) & $2.40\times 10^{-4}$ ($0.13\%$) & $148$ & $559$ & $45$ \\
 	    & $4$ & $0.0965$ ($0.08\%$) & $2.31\times 10^{-5}$ ($0.08\%$) & $163$ & $10,108$ & $772$ \\
 	    & $5$ & $0.5162$ ($0.04\%$) & $3.08\times 10^{-6}$ (0.04\%) & $116$ & $220,943$ & $24,978$ \\
		\bottomrule
	\end{tabular}
\end{table}

\begin{table}[htbp]\small
	\centering
	\caption{Numerical results for estimating portfolio loss probabilities in Example~\ref{ex:portfolio}}
    \label{tab:portfolio}
	\begin{tabular}{ccccccc}\toprule
		Method&$r$&$u_r$ (95\% rel.~err.)&$p_r$ (95\% rel.~err.)&Time (sec)& $\VR$ &$\ER$\\ \midrule
		\multirow{3}*{MC}
		& $2$ & $1.42$ ($1.865\%$) & $1.05\times 10^{-4}$ ($2.499\%$) & $7$ & -- & -- \\
   	  & $3$ & $8.46$ ($2.278\%$) & $1.37\times 10^{-5}$ ($3.482\%$) & $7$ & -- & -- \\
   	  & $4$ & $24.50$ ($2.512\%$) & $4.40\times 10^{-6}$ ($3.362\%$) & $7$ & -- & -- \\
		\midrule
		\multirow{3}*{ET}
		& $2$ & $1.40$ ($0.056\%$) & $1.05\times 10^{-4}$ ($0.042\%$) & $48$ & $3,615$ & $526$ \\
   	  & $3$ & $8.60$ ($0.023\%$) & $1.35\times 10^{-5}$ ($0.016\%$) & $50$ & $48,120$ & $6,620$ \\
   	  & $4$ & $24.73$ ($0.013\%$) & $4.39\times 10^{-6}$ ($0.009\%$) & $51$ & $145,230$ & $19,182$ \\
		\midrule
		\multirow{3}*{DRIS}
		& $2$ & $1.40$ ($0.024\%$) & $1.05\times 10^{-4}$ ($0.034\%$) & $53$ & $5,269$ & $691$ \\
   	  & $3$ & $8.60$ ($0.009\%$) & $1.35\times 10^{-5}$ ($0.013\%$) & $54$ & $71,806$ & $9,212$ \\
   	  & $4$ & $24.73$ ($0.004\%$) & $4.39\times 10^{-6}$ ($0.007\%$) & $51$ & $227,647$ & $30,143$ \\
		\bottomrule
	\end{tabular}
\end{table}

\section{Concluding Remarks}\label{sec:conclude}
In this paper, we address the problem of efficiently estimating rare-event probabilities under distributional model risk. Leveraging strong duality results in Wasserstein distributionally robust optimization, we formulate a novel, computationally tractable importance sampling procedure called DRIS, which yields significant variance reduction in estimating the said probabilities. We rigorously prove that the proposed DRIS estimator achieves vanishing relative error, which is regarded as the strongest notion of efficiency in the context of rare-event simulation. All our numerical experiments support these theoretical findings.

As the first methodological framework specifically designed to estimate rare-event probabilities under distributional uncertainty, our proposed approach relies on specific modeling assumptions that suggest several interesting avenues for future research. Firstly, we focus on convex sets as target events, motivated by several examples in the relevant literature. Nevertheless, extending our methodology to non-convex target sets, while challenging, would substantially expand its practical applicability. Secondly, we restrict our focus to the case with Gaussian nominal distributions. While the framework extends to other elliptical nominal distributions as alluded to earlier, the theoretical performance in those cases remains to be verified. It would also be interesting to explore the cases with non-elliptical nominal distributions. Lastly, to ensure the tractability of our theoretical analysis, we use the 2-Wasserstein ball to define the distributional uncertainty set. Relaxing this constraint would be a promising direction, as the duality result in Lemma~\ref{lem:dual} generalizes to a broader class of uncertainty sets, including $p$-Wasserstein balls with $p\geq1$. 
\appendix
\section{Proofs of the Theoretical Results}\label{app:main}

\begin{proof}[{Proof of Lemma \ref{lem:u}.}]
	Fix $K>M>0$. Assume by contradiction that $u_r\geq r-M$ for some $r>M$.  Then, we observe that
	\begin{equation}\label{eq:contradict}
    \begin{aligned}
		\frac{\delta^2}{r^2}=\frac{h_r(u_r)}{r^2}
		&\geq\sE\lt[\frac{d(\bX,\set_r)^2}{r^2};d(\bX,\set_r)\leq r-M,\|\bX\|\leq K\rt].
	\end{aligned}
	\end{equation}
	Since $d(\cdot,\set_r)$ is $1$-Lipschitz, we have $d(\bx,\set_r)\geq d(\bzero,\set_r)-\|\bx\|\geq r- K$ for any $\bx$ satisfying $\|\bx\|\leq K$.  
    This implies that $\liminf_{r\to\infty}{d(\bx,\set_r)^2}/{r^2}=1$.

    Fix $\bx\in\bbR^n$ such that $x_1>M$. Then, 
	by letting $t_r\coloneqq{(\|\bx\|^2-M^2)}/{(2rx_1-2rM)}>0$, a straightforward calculation yields
    $\|\bx-rt_r\be_1\|=rt_r-M$.
    Thus, since $d(\cdot,\set_r)$ is $1$-Lipschitz, we have
    $d(\bx,\set_r)\leq d(rt_r\be_1,\set_r)+\|\bx-rt_r\be_1\|=r-rt_r+rt_r-M=r-M$ for all sufficiently large $r$ such that $t_r\in(0,1)$.
    Accordingly, by applying Fatou's lemma on \eqref{eq:contradict}, we obtain
	$$\liminf_{r\to\infty}\frac{\delta^2}{r^2}\geq \sE_0\lt[\liminf_{r\to\infty}\frac{d(\bX,\set_r)^2}{r^2};d(\bX,\set_r)\leq r-M,\|\bX\|\leq K\rt]\geq \sP(X_1>M,\|\bX\|\leq K)>0.$$
	This contradicts the fact that $\delta$ is a constant. Therefore, $u_r<r-M$ for all sufficiently large $r$. 

    Furthermore, it is straightforward that $p_r=\sP(d(\bX,\set_r)\leq u_r)\leq \sP(X_1\geq r-u_r) = \bar\Phi(r-u_r)$. Hence, we get $\delta^2=h_r(u_r)=\sE[d(\bX,\set_r)^2;d(\bX,\set_r)\leq u_r]\leq u_r^2p_r\leq r^2\bar\Phi(r-u_r)$ for all sufficiently large $r$. Consequently, the result follows.	
\end{proof}

\begin{proof}[{Proof of Theorem \ref{thm:prob-asymp}.}] 
By the asymptotic behavior of the Mills ratio for a standard normal distribution, we have $\sqrt{2\pi}x\bar\Phi(x)/\exp(-x^2/2)\to1$ as $x\to\infty$~\citep[see, e.g.,][]{Bartoszynski2021}. This implies that $x^2\bar\Phi(x)\to0$ as $x\to\infty$. Thus, by letting $x=\bar\Phi^{-1}(\delta^2/r^2)$, we have $r^{-1}\bar\Phi^{-1}(\delta^2/r^2)\to0$ as $r$ grows. Then, dividing both sides of \eqref{eq:u-bound} by $r$ and letting $r\to\infty$ yields $\lim_{r\to\infty}u_r/r=1$.
Furthermore, we observe that $$
h_r(u_r)=\sE_0[d(\bX,\set_r)^2;d(\bX,\set_r)\leq u_r]\leq u_r^2\sP_0(d(\bX,\set_r)\leq u_r)=u_r^2p_r.
$$ 
Consequently, $\liminf_{r\to\infty}r^2p_r\geq \delta^2/\lim_{r\to\infty}(u_r/r)^2=\delta^2$.
\end{proof}

\begin{proof}[{Proof of Theorem \ref{thm:clt}.}]	We prove the statement in four steps. In this proof, we denote by $\|\cdot\|_2$ the $L^2$ norm under the sampling distribution, i.e., $\|A\|_2 = \sqrt{\sE[A(\bZ)^2]}$ for any function $A:\bbR^n\to\bbR$. 

	\emph{Step 1: Uniform Convergence of $\widehat h_{N}$.} 
    In this step, we aim to prove the uniform convergence of $\widehat h_{N}$ in \eqref{eq:h_hat_N} over $\Theta\coloneqq[u_L,u_U]$. 
    Since every Donsker class satisfies the uniform law of large numbers~\citep[][page 130]{van1996weak}, it suffices to show that $\cH\coloneqq\{H(\cdot,u):u\in \Theta\}$ is Donsker.

    We define two function classes $\cH_1$ and $\cH_2$ as $\cH_1\coloneqq\{\bz\mapsto \lt(d(f_u(\bz),\set)\wedge u_U\rt)^2L_u(\bz):u\in\Theta\}$ and $\cH_2\coloneqq\{\bz\mapsto\sone\{d(f_u(\bz),\set)\leq u\}:u\in\Theta\}$. We observe that for any $u,v\in\Theta$,
    \begin{equation}
    \begin{aligned}
    &\lt|\lt(d(f_u(\bz),\set)\wedge u_U\rt)^2L_u(\bz)-\lt(d(f_v(\bz),\set)\wedge u_U\rt)^2L_v(\bz)\rt|\\
    &\leq\lt(d(f_u(\bz),\set)\wedge u_U\rt)^2\lt|L_u(\bz)-L_v(\bz)\rt|+|L_v(\bz)|\lt|\lt(d(f_u(\bz),\set)\wedge u_U\rt)^2-\lt(d(f_v(\bz),\set)\wedge u_U\rt)^2\rt|\\
    &\leq u_U^2\lt|L_u(\bz)-L_v(\bz)\rt|+\bar L\lt|\lt(d(f_u(\bz),\set)\wedge u_U\rt)^2-\lt(d(f_v(\bz),\set)\wedge u_U\rt)^2\rt|\\
    &\leq u_U^2\lt|L_u(\bz)-L_v(\bz)\rt|+2u_U\bar L\lt\|f_u(\bz)-f_v(\bz)\rt\|,
    \end{aligned}
    \end{equation}
    where the first inequality follows from the triangular inequality, the second inequality holds since $\bar L \coloneqq \sup_{\bz\in\bbR^n,u\in\Theta} L_u(\bz)<\infty$, and the last one is straightforward because $|a^2-b^2|\leq2c|a-b|$ for $a,b\in[0,c]$ and $c\geq0$, and $d(\cdot,\set)$ is $1$-Lipschitz. It can be easily checked that there exists a polynomial function $G$ satisfying $u_U^2\lt|L_u(\bz)-L_v(\bz)\rt|+2u_U\bar L\lt\|f_u(\bz)-f_v(\bz)\rt\|\leq G(\bz)|u-v|$ for all $\bz\in\bbR^n$ and $u,v\in\Theta$. Since $\|G\|_2<\infty$ and $\Theta$ is compact, $\cH_1$ is Donsker by Theorems 2.7.17 and 2.5.6 of \cite{van1996weak}.

    Given a collection $\cC$ of sets, its {VC-dimension}, denoted by $V(\cC)$, is the cardinality of the largest set $X$ such that $\lt|\{X\cap C:C\in\cC\}\rt|=2^{|X|}$. A function class $\cF$ is called a {VC-class} if the collection of all subgraphs
	$\{\{(\bz,t):t<f(\bz)\}:f\in\cF\}$ has a finite VC-dimension. 
    Suppose that $|\{\{(\bz_1, t_1), \dots, (\bz_m, t_m)\}\cap\{(\bz,t):t<\sone\{d(f_u(\bz),\set)\leq u\}\}:u\in\Theta\}| = 2^m$ for some $m$ points $(\bz_1, t_1), \dots, (\bz_m, t_m)\in(0,\infty)\times\bbR^{n-1}\times\bbR$. Since the condition $t<\sone\{d(f_u(\bz),\set)\leq u\}$ is nontrivial only when $t\in[0,1)$, we may choose $t_1 = \cdots = t_m = 0$ without loss of generality. In this case, the shattering condition on subgraphs is equivalent to shattering the points $\bz_1, \dots, \bz_m$ directly using the function values, i.e., $\lt|\{(\sone\{d(f_u(\bz_1),\set)\le u\}, \ldots, \sone\{d(f_u(\bz_m),\set)\le u\}):u\in\Theta\}\rt|=2^m$.

    On the other hand, by Lemma~\ref{lem:vc-class} in Appendix~\ref{app:lemma}, the set $\{u \in \Theta : d(f_u(\bz_i), \set) \le u\}$ is defined by at most 2 boundary points in $\Theta$. Hence, there exist at most $2m$ points in $\Theta$, denoted by $u_1,u_2,\ldots,u_{2m}$, such that $u_L=u_0\leq u_1\leq\cdots\leq u_{2m}\leq u_{2m+1}=u_U$ and the vector $(\sone\{d(f_u(\bz_1),\set)^2\le u\}, \ldots, \sone\{d(f_u(\bz_m),\set)^2\le u\})$ remains constant for any $u\in(u_i,u_{i+1})$ with $i=0,\ldots,2m$. Thus, $\lt|\{(\sone\{d(f_u(\bz_1),\set)\le u\}, \ldots, \sone\{d(f_u(\bz_m),\set)\le u\}):u\in\Theta\}\rt|\leq  2m+1$. Combining this with the above shattering condition leads to $2^m\leq 2m+1$. Therefore, $m$ must be finite, proving that $\cH_2$ is a VC-class. Furthermore, $\cH_2$ is uniformly bounded by $1$. Consequently, Theorems~2.6.7 and~2.5.2 of \cite{van1996weak} imply that $\cH_2$ is Donsker.

    Let $\phi(x,y)=xy$ for all $x,y\in\bbR$. Since $\cH_1$ and $\cH_2$ are uniformly bounded and Donsker and $\cH\subset\phi\circ(\cH_1,\cH_2)\coloneqq\{\bz\mapsto \phi(g_1(\bz),g_2(\bz)):g_1\in\cH_1,g_2\in\cH_2\}$, $\cH$ is also Donsker by Corollary~2.10.15 and Theorem~2.10.1 of \cite{van1996weak}.
    
    \emph{Step 2. Convergence of $\widehat u_N$.} Since $h(\cdot)$ is a strictly increasing function satisfying $h(u_*)=\delta^2$, we have $c(\varepsilon)\coloneqq\inf_{|u-{u_*}|>\varepsilon}|h(u)-\delta^2|/2>0$ for any $\varepsilon>0$. Fix $\varepsilon>0$. If $\sup_{u\in\Theta}|h(u)-\widehat h_N(u)|\leq  c(\varepsilon)$, then $|h(\widehat u_N)-\delta^2|\leq \max\{\lim_{u\uparrow\widehat u_N}|h(u)-\widehat h_N(u)|,\lim_{u\downarrow\widehat u_N}|h(u)-\widehat h_N(u)|\}\leq c(\varepsilon)$,
	which implies that $|\widehat u_N-{u_*}|\leq \varepsilon$. Accordingly,
	$\sP(\sup_{u\in\Theta}|h(u)-\widehat h_N(u)|\leq c(\varepsilon))\leq \sP(|\widehat u_N-{u_*}|\leq \varepsilon)$. By the uniform convergence of $\widehat h_N$ in Step~1, $\lim_{N\to\infty}\sP(\sup_{u\in\Theta}|h(u)-\widehat h_N(u)|\leq c(\varepsilon))=1$.
	Hence, $\widehat u_N\rightarrow {u_*}$ in probability as $N\to\infty$. 
	
	\emph{Step 3. Asymptotic Normality for $\widehat u_N$.} 
    We define $H_1(\bz,u)=(d(f_u(\bz),\set)\wedge u_U)^2L_u(\bz)$ and $H_2(\bz,u) = \sone\{d(f_u(\bz),\set)\leq u\}$, implying that 
	$H(\bz,u) = H_1(\bz,u)H_2(\bz,u)$  for $\bz\in(0,\infty)\times\bbR^{n-1}$ and $u\in\Theta$.
    We observe that $d(f_{{u_*}}(\bz),\set)={u_*}$ if and only if $f_{{u_*}}(\bz)$ lies on the boundary of $\{\bz:d(\bz,\set)\leq {u_*}\}$. Additionally, since $f_{{u_*}}$ is an invertible affine transformation, it can be checked that $\sP(d(f_{{u_*}}(\bZ),\set)={u_*})=0$. 
    
    Fix $\omega$ in the sample space such that $d(f_{{u_*}}(\bZ(\omega)),\set)\neq{u_*}$. Then, since $u\mapsto d(f_{u}(\bZ(\omega)),\set) -u$ is continuous, there exists $\delta>0$ such that $H_2(\bZ(\omega),u)=H_2(\bZ(\omega),{u_*})$ for any $|u-{u_*}|<\delta$. Therefore, $H_2(\bZ,u)\to H_2(\bZ,{u_*})$ almost surely as $u\rightarrow {u_*}$. 
    Thus, by the continuity of $H_1(\bz,\cdot)$ and the continuous mapping theorem, $\|H(\cdot,u)-H(\cdot,{u_*})\|_2^2=\|H_1(\cdot,u)H_2(\cdot,u)-H_1(\cdot,{u_*})H_2(\cdot,{u_*})\|_2^2\to0$
    as $u\rightarrow {u_*}$. 
    We also note that $\{H(\cdot,u)-H(\cdot,{u_*}): |u-{u_*}|<\delta, u\in\Theta\}$ is Donsker for some $\delta>0$ since $\cH$ is Donsker and by Theorem~2.10.8 of \cite{van1996weak}.

    Let $\Psi_N(u)\coloneqq\widehat h_N(u)-\delta^2$ and $\Psi(u)\coloneqq h(u)-\delta^2$. Then, by the central limit theorem, we have
	$\sqrt{N}(\Psi_N-\Psi)({u_*})=N^{-1/2}\sum_{i=1}^N(H(\bZ_i,{u_*})-\sE[H(\bZ,{u_*})])\Rightarrow \cN(0,\Var(H(\bZ,{u_*}))).$ Furthermore, $\Psi'({u_*})=h'({u_*})\neq 0$ by Lemma~\ref{lem: differentiability} in Appendix~\ref{app:lemma}. Moreover, since $H(\bz,u)$ is uniformly bounded, it can be verified that $\Psi_N(\widehat u_N)=o_P(N^{-1/2})$ using the definition of $\widehat u_N$ and $\bZ_i$ is continuously distributed. Combining all these results with Lemma~3.3.5 and Theorem~3.3.1 of \cite{van1996weak}, we conclude that $\sqrt{N}h'({u_*})(\widehat u_N-{u_*})=-\sqrt{N}(\widehat h_N-h)({u_*})+o_P(1)$.

	\emph{Step 4. Asymptotic Normality for the Estimator.} Using the same arguments as in Steps 1 to 3, it can be shown that $\{P(\cdot,u)-P(\cdot,{u_*}):|u-{u_*}|<\delta, u\in\Theta\}$ is Donsker for some $\delta>0$, and $\|P(\cdot,u)-P(\cdot,{u_*})\|_2^2\rightarrow 0$ as $u\rightarrow{u_*}$. Thus, by using Lemma 3.3.5 of \cite{van1996weak} again, we have $\sqrt{N}\lt(\widehat p_N(\widehat u_N)-p(\widehat u_N)\rt)=\sqrt{N}\lt(\widehat p_N({u_*})-p({u_*})\rt)+o_P(1)$.
	Since $p(\cdot)$ is differentiable at ${u_*}$, the Taylor expansion implies that $\sqrt{N}(p(\widehat u_{N})-p({u_*})) = \sqrt{N}p'({u_*})(\widehat u_N-{u_*})+o_P(\sqrt{N}|\widehat u_N-{u_*}|)$.
	Combining these findings with the result of Step~3 and Lemma~\ref{lem: differentiability} in Appendix~\ref{app:lemma},  we obtain
	\begin{equation}
	\begin{aligned}
		\sqrt{N}(\widehat p_N(\widehat u_N)-p({u_*}))
		& = \sqrt{N}(\widehat p_N-p)({u_*})+\sqrt{N}p'({u_*})(\widehat u_N-{u_*})+o_P(\sqrt{N}|\widehat u_N-{u_*}|)+o_P(1)\\
		& = \sqrt{N}(\widehat p_N-p)({u_*})-\frac{\sqrt{N}}{u_*^2}(\widehat h_N-h)({u_*})+o_P(1),
	\end{aligned}
	\end{equation}
	where the last equality holds since $\sqrt{N}(\widehat u_N-{u_*})$ is bounded in probability by Step 3. Hence, by the central limit theorem and Slutsky's theorem, the desired result in \eqref{eq:CLT-var} follows.
\end{proof}
\begin{proof}[{Proof of Theorem \ref{thm:bre}.}]
	Since $\bx^*=\argmin_{\bx\in\set_r}\|\bx\|^2=r\be_1$ and $x_1^*=r$, we observe that \begin{equation}\label{eq:var-upper}
	\begin{aligned}
		\sigma_r^2
		&=\Var\lt(\sone\{d(f_{u_r}(\bZ),\set_r)\leq u_r\}L_{u_r}(\bZ)\lt(1-\frac{d(f_{u_r}(\bZ),\set_r)^2}{u_r^2}\rt)\rt)\\
		&\leq \sE\lt[L_{u_r}(\bZ)^2\lt(1-\frac{d(f_{u_r}(\bZ),\set_r)^2}{u_r^2}\rt)^2;d(f_{u_r}(\bZ),\set_r)\leq u_r\rt]\\
		&=\sE_0\Bigg[\ell_{u_r}(X_1) \lt(1-\frac{d(\bX,\set_r)^2}{u_r^2}\rt)^2;d(\bX,\set_r)\leq u_r\Bigg]\\
		&\leq \lt(\sE_0\big[\ell_{u_r}(X_1)^2:d(\bX,\set_r)\leq u_r\big]\sE_0\Bigg[\lt(1-\frac{d(\bX,\set_r)^2}{u_r^2}\rt)^4:d(\bX,\set_r)\leq u_r\Bigg]\rt)^{1/2},
	\end{aligned}
	\end{equation}
	where $\ell_u(x)\coloneqq e^{-x^2/2+(r-u)(x-(r-u))}/((r-u)\sqrt{2\pi})\sone{\{x\geq r-u\}}$ and the last inequality holds by the Cauchy–Schwarz inequality. 
    A simple calculation yields
	\begin{equation}
		\label{eq:upper-bound-2}
		\sE_0\big[\ell_{u_r}(X_1)^2;d(\bX,\set_r)\leq u_r\big]\leq \sE_0[\ell_{u_r}(X_1)^2;X_1\geq r- u_r]\leq\frac{e^{-3(r-u_r)^2/2}}{(2\pi)^{3/2}(r-u_r)^3}.
	\end{equation}
	Also, by using integration by parts, we have
    \begin{equation}\label{eq:upper-bound-3}
	\begin{aligned}
		\sE_0\lt[\lt(1-\frac{d(\bX,\set_r)^2}{u_r^2}\rt)^4:d(\bX,\set_r)\leq u_r\rt]
		&\leq \frac{8}{u_r^2}\int_{0}^{u_r}t\lt(1-\frac{t^2}{u_r^2}\rt)^3\sP_0(d(\bX,\set_r)\leq t)\rd t\\
		&\leq\frac{8}{u_r^2}\int_0^{u_r}t\lt(1-\frac{t^2}{u_r^2}\rt)^3\frac{e^{-(r-t)^2/2}}{\sqrt{2\pi}(r- t)}\rd t\\
		&\leq\frac{64e^{-(r-u_r)^2/2}}{\sqrt{2\pi}(r- u_r)u_r^4}\int_0^{u_r}(u_r-t)^3e^{-(r-u_r)(u_r-t)}\rd t\\
		&\leq\frac{64e^{-(r-u_r)^2/2}}{\sqrt{2\pi}(r- u_r)u_r^4}\int_0^{\infty}y^3e^{-(r-u_r)y}\rd y\\
		&\leq\frac{384e^{-(r-u_r)^2/2}}{\sqrt{2\pi}(r- u_r)^5u_r^4},
	\end{aligned}
	\end{equation}
    where the second inequality holds since $\sP_0(d(\bX,\set_r)\leq t)\leq \bar\Phi(r-t) \leq (2\pi)^{-1/2}e^{-(r-t)^2/2}/(r-t)$ for any $t\in[0,r]$, and the third inequality follows because $t(1-t^2/u_r^2)^3\leq 8(u_r-t)^3/u_r^2$ and $e^{-(r-t)^2/2}/(r-t)\leq e^{-(r-u_r)^2/2-(r-u_r)(u_r-t)}/(r-u_r)$ for all $t\in[0,u_r]$.
	By \eqref{eq:var-upper}, \eqref{eq:upper-bound-2}, and \eqref{eq:upper-bound-3}, we have
	\begin{equation}\label{eq:var-bound}
	\begin{aligned}
		\sigma_r^2\leq\lt(\frac{e^{-3(r-u_r)^2/2}}{(2\pi)^{3/2}(r-u_r)^3}\frac{384e^{-(r-u_r)^2/2}}{\sqrt{2\pi}(r- u_r)^5u_r^4}\rt)^{1/2}=\frac{4\sqrt 6e^{-(r-u_r)^2}}{\pi(r-u_r)^4u_r^2}.
	\end{aligned}
	\end{equation}

    Suppose that $n\geq 2$. Fix $w,u>0$ satisfying $w<u^2$. 
    Assume that $r-u<x_1< r-\sqrt{w}$ and $\|\bx-r\be_1\|\leq u$ for some $\bx\in\bbR^n$. Since  $r\be_1\in\set_r$, we have $d(\bx,\set_r)\leq u$. Let $\bar \bx = \argmin_{\by\in\set_r} \|\bx-\by\|$. Then, $\bar x_1\geq r$, and thus, $d(\bx,\set_r)\geq \bar x_1-x_1>\sqrt{w}$. Furthermore,  $\sP_0(\|\bX-r\be_1\|\leq u\,|\,X_1=x)$ is equal to the probability of a chi-squared random variable with $n-1$ degree of freedom not exceeding $u^2-(x-r)^2$ since $\|\bx-r\be_1\|^2 = (x_1-r)^2 + \sum_{i=2}^nx_i^2$ for any $\bx\in\bbR^n$. Accordingly, there exists $C>0$ such that
	\begin{equation}\label{eq:p-band-2}
	\begin{aligned}
    \sP_0(w<d(\bX,\set_r)^2\leq u^2)&\geq \sP_0(r-u<X_1<r-\sqrt{w},\|\bX-r\be_1\|\leq u)\\
		&=\int_{r-u}^{r-\sqrt{w}}\sP_0(\|\bX-r\be_1\|\leq u\,|\,X_1=x)\frac{e^{-x^2/2}}{\sqrt{2\pi}}\rd x\\
		&=C\int_{r-u}^{r-\sqrt{w}}\int_0^{u^2-(x-r)^2}t^{(n-3)/2}e^{-t/2}e^{-x^2/2}\rd t\rd x.
	\end{aligned}
	\end{equation}
	Using integration by parts, one can show that $h_r(u)=\int_0^{u^2}\sP(w<d(\bX,\set_r)^2\leq u^2)\rd w$. Then, by \eqref{eq:p-band-2},  
	\begin{equation}
	\begin{aligned}
    h_r(u)&\geq C\int_0^{u^2}\int_{r-u}^{r-\sqrt{w}}\int_0^{u^2-(x-r)^2}t^{(n-3)/2}e^{-t/2}e^{-x^2/2}\rd t\rd x\rd w\\		&=2C\int_0^u\int_0^{y^2}\int_0^{(u^2-y^2)^{1/2}}s^{n-2}e^{-s^2/2}e^{-(r-y)^2/2}\rd s\rd w \rd y\\
    &=2C\int_0^{u} y^2\int_0^{(u^2-y^2)^{1/2}}s^{n-2}e^{-s^2/2}e^{-(r-y)^2/2}\rd s\rd y\\
	&=2C\int_0^{u}\rho^{n+1}e^{-(r-\rho)^2/2}\int_0^{\pi/2}\cos(\theta)^2\sin(\theta)^{n-2}e^{-r\rho(1-\cos(\theta))}\rd\theta \rd \rho,\label{eq:polar}
	\end{aligned}
	\end{equation}
    where the first equality holds by interchanging the first two integrals and setting $s=\sqrt{t}$ and $y=r-x$, and the last equality follows from setting $s = \rho\sin(\theta)$ and $y = \rho\cos(\theta)$.
	
    Let $\varepsilon_r = 1/u_r$. Then, $0\leq\varepsilon_r/r\leq 1$ for all sufficiently large $r$. Thus, for all $\rho\in(0,u)$, the inner integral of the last expression in \eqref{eq:polar} satisfies
	\begin{equation}
	\begin{aligned}
		\int_0^{\pi/2}\cos(\theta)^2\sin(\theta)^{n-2}e^{-r\rho(1-\cos(\theta))}\rd\theta
		&\geq\int_0^{\arccos(1-\varepsilon_r/r)}\cos(\theta)^2\sin(\theta)^{n-2}e^{-r\rho(1-\cos(\theta))}\rd\theta\\
		&\geq\lt(1-\frac{\varepsilon_r}{r}\rt)^2e^{-\varepsilon _r\rho}\int_0^{\arccos(1-\varepsilon_r/r)}\sin(\theta)^{n-2}\rd\theta\\
		&=\lt(1-\frac{\varepsilon_r}{r}\rt)^2e^{-\varepsilon_r \rho}\int_0^{\varepsilon_r/r}{(2\alpha)}^{(n-3)/2}{(1-\alpha/2)}^{(n-3)/2}\rd \alpha\\
		&\geq\kappa_re^{-\varepsilon_r \rho},\label{eq:polar-low}
	\end{aligned}
	\end{equation}
    where $\kappa_r = (1-{\varepsilon_r}/{r})^2({1-{\varepsilon_r}/{(2r)}})^{(n-3)/2}({2\varepsilon_r}/{r})^{(n-1)/2}/(n-1)$, and  the equality stems from setting $\theta = \arccos(1-\alpha)$. 
	Hence, by \eqref{eq:polar} and using integration by parts twice, we obtain
	\begin{equation}
		\begin{aligned}
			h_r(u_r)&\geq 2C\kappa_r\int_0^{u_r}\rho^{n+1}e^{-(r-\rho)^2/2-\varepsilon_r\rho}\rd \rho\\		&=2C\kappa_r\lt(I_r(u_r)+\int_0^{u_r}\frac{e^{-(r-\rho)^2/2-\varepsilon_r\rho}}{(r-\varepsilon_r-\rho)^2}\rho^{n-1}\lt(n(n+1)+\frac{3(n+1)\rho}{r-\varepsilon_r-\rho}+\frac{3\rho^2}{(r-\varepsilon_r-\rho)^2}\rt)\rd\rho\rt)\\
			&\geq 2C\kappa_rI_r(u_r),
		\end{aligned}
	\end{equation}
    where $$
    I_r(u_r)\coloneqq\frac{e^{-(r-u_r)^2/2-1}}{r-u_r^{-1}-u_r}u_r^{n+1}\lt(1-\frac{n+1}{u_r(r-u_r^{-1}-u_r)}-\frac{1}{(r-u_r^{-1}-u_r)^2}\rt).
    $$

    Recall that $r-u_r\to\infty$ and $u_r/r\to1$ as $r\to\infty$ by Lemma~\ref{lem:u} and the proof of Theorem~\ref{thm:prob-asymp}. Thus, we have $\kappa_r\sim r^{-n+1}2^{(n-1)/2}/(n-1)$ and $I_r(u_r)\sim{r^{n+1}e^{-(r-u_r)^2/2-1}}/({r-u_r^{-1}-u_r})$, where $\sim$ represents asymptotic equivalence  as $r\to\infty$.
    Since $\delta^2=h_r(u_r)$ for all $r$, the above inequality implies that
    \begin{equation}
	\begin{aligned}
		\limsup_{r\to\infty}r^2\frac{e^{-(r-u_r)^2/2-1}}{r-u_r^{-1}-u_r}\leq C_*\delta^2,
	\end{aligned}
	\end{equation}
    where $C_*=(n-1)/({2^{(n+1)/2}C})$. Finally, combining this result with \eqref{eq:var-bound} and Theorem \ref{thm:prob-asymp}, we get
	\begin{equation}\label{eq:vre_proof}
	\begin{aligned}
		\limsup_{r\to\infty}r^2(r-u_r)^2\frac{\sigma_r^2}{p_r^2}&\leq\frac{4\sqrt 6}{\pi}\frac{1}{\liminf_{r\to\infty}r^4p_r^2}\limsup_{r\to\infty}r^4\frac{e^{-(r-u_r)^2}}{(r-u_r)^2}\\
		&\leq\frac{4e^2\sqrt 6}{\pi\delta^4}\lt(\limsup_{r\to\infty}r^2\frac{e^{-(r-u_r)^2/2-1}}{r-u_r^{-1}-u_r}\rt)^2\\
		&\leq \frac{4C_*^2e^2\sqrt 6}{\pi}<\infty.
	\end{aligned}
	\end{equation}
    
	When $n=1$, we obtain the following relationship using the same argument as in \eqref{eq:polar}:
    $$
    \begin{aligned}
    h_r(u_r)&=\int_0^{u_r^2}\sP(w<d(\bX,\set_r)^2\leq u_r^2)\rd z\\
		&=(2\pi)^{-1/2}\int_0^{u_r^2}\int_{r-u_r}^{r-\sqrt{w}}e^{-x^2/2}\rd x\rd w\\
        &=(2\pi)^{-1/2}\int_0^{u_r}y^2e^{-(r-y)^2/2}dy.
    \end{aligned}
        $$
    Using integration by parts, the right-hand side is bounded from below by
    $$
    \frac{1}{\sqrt{2\pi}}\frac{e^{-(r-u_r)^2/2}}{r-u_r}\lt(\frac{u_r}{r}\rt)^2r^2\lt(1-\frac{2}{u_r(r-u_r)}-\frac{1}{(r-u_r)^2}\rt)\sim \frac{r^2e^{-(r-u_r)^2/2}}{\sqrt{2\pi}(r-u_r)}.
    $$
	 Analogous to \eqref{eq:vre_proof}, we apply $\delta^2=h_r(u_r)$ and arrive at
	\begin{equation}
	\begin{aligned}
		\limsup_{r\to\infty}r^2(r-u_r)^2\frac{\sigma_r^2}{p_r^2}\leq&\frac{4\sqrt 6}{\pi} \frac{1}{\liminf_{r\to\infty}r^4p_r^2}\limsup_{r\to\infty}r^4\frac{e^{-(r-u_r)^2}}{(r-u_r)^2}\leq8\sqrt 6<\infty.
	\end{aligned}
	\end{equation}
	This completes the proof.
\end{proof}

\section{Technical Lemmas}\label{app:lemma}

\begin{lemma}\label{lem:vc-class}
	Fix $\bz\in(0,\infty)\times\bbR^{n-1}$ and let $g(u)=d(f_u(\bz),\set)-u$ for any $u$ in a compact interval $\Theta$ of $(0,x_1^*)$. We say that $v$ is a zero-crossing if it is in the interior of $\Theta$ and there exists $\delta>0$ such that $\sone\{g(v-t)\leq 0\}\neq \sone\{g(v+t)\leq 0\}$  for all $t \in(0,\delta)$. Then, there are at most two zero-crossings.
\end{lemma}
\begin{proof}
	Let $z_1(u)=x_1^*-u+z_1/(x_1^*-u)$ be the first coordinate of $f_u(\bz)$ for $u\in\Theta$. We write $\Theta =[u_L,u_U]$ for some $0<u_L\leq u_U<x_1^*$.
    
    Let $u_* = \min\{\max\{u_L, x_1^*-\sqrt{z_1}\},u_U\}$, $I_1\coloneqq[u_L,u_*]$, and $I_2\coloneqq (u_*,u_U]$. Since $d(\cdot,\set)$ is 1-Lipschitz and $|z_1'(\cdot)|\leq 1$ on $I_1$, we have $|g(u)+u-g(v)-v|\leq |z_1(u)-z_1(v)|\leq|u-v|$, which implies that $g(u) \leq g(v) + v-u+|u-v|$
	for any $u,v\in I_1$. Thus, if  $g(v)\leq 0$ for some $v\in I_1$, then $g(u) \leq 0$ for all $u\in[v,u_*]$.
    
    On the other hand, $z_1(\cdot)$ is strictly increasing and convex on $I_2$. Moreover, it can be easily verified that $d(\by,\set)$ is convex in $y_1$.
    Thus, $d(f_u(\bz),\set)$ is decreasing with respect to $u$ on $(u_*,w]$ and increasing on $(w,u_U]$ for some $w\in(u_*,u_U]$. This suggests that $g(\cdot)$ is also decreasing on $(u_*,w]$. Furthermore, $g(\cdot)$ is convex on $(w,u_U]$. 
	Therefore, there are at most two zero-crossings in $\Theta$.
\end{proof}  

\begin{lemma}
	\label{lem: differentiability}
	$h(\cdot)$ and $p(\cdot)$ are differentiable at ${u_*}$ with $h'({u_*})=u_*^2p'({u_*})\neq0$. 
\end{lemma}
\begin{proof}
Recall that $p(u) = \sP(d(\bX,\set)\leq u)$.
It is straightforward to check that $d(\cdot,\set)$ is 1-Lipschitz and differentiable almost everywhere with $\|\nabla d(\cdot,\set)\|=1$. Then, by the coarea formula \cite[Theorem 3.4.2]{EvansGariepy1992}, we have
\begin{equation}\label{eq:coarea}
	\begin{aligned}
		p({u_*}+\delta) - p({u_*}) &=
		\sP({u_*}<d(\bX,\set)\leq {u_*}+\delta) \\
		&= \int_{\bbR^d} \phi(\bx)\sone{\{{u_*}<d(\bx,\set)\leq {u_*}+\delta\}}\|\nabla d(\bx,\set)\|\rd\bx\\
		&=\int_\bbR\lt(\int_{\partial(\set+B({u_*}+t))}\phi(\bz)\sone{\{{u_*}<d(\bz,\set)\leq{u_*}+\delta\}}\rd\cH(\bz)\rt)\rd t\\
		&=\int_0^\delta\lt(\int_{\partial(\set+B({u_*}+t))}\phi(\bz)\rd\cH(\bz)\rt)\rd t,
	\end{aligned}
\end{equation}
where $\phi(\bz) = (2\pi)^{-n/2} e^{-\|\bz\|^2/2}$ is the density of the $n$-dimensional standard Gaussian distribution, and $\cH$ is the $(n-1)$-dimensional Hausdorff measure.

We write $\set_u\coloneqq\{\bx:d(\bx,\set)\leq u\}$. By the fundamental theorem of calculus, it suffices to show that 
$g(u)\coloneqq\int_{\partial\set_u}\phi(\bz)\rd\cH(\bz)$
is continuous on $(0,\infty)$. 
To that end, we fix $u>0$ arbitrarily and denote by $n(\bz)$ the outer unit normal vector at $\bz \in\partial\set_u$. Then, by the change of variables, 
$$ 
g(u+t) = \int_{\partial\set_u} \phi(\bz + t n(\bz)) J_t(\bz) \rd\cH(\bz),
$$
where $J_t(\bz)$ denotes the Jacobian of the mapping $\bz\mapsto \bz + t n(\bz)$ for each $t\geq 0$. 
By the smoothness of $\partial\set_u$ and the convexity of $\set_u$, it is not difficult to check that the Jacobian $J_t(\bz)$ is nonnegative and continuous in both $\bz$ and $t$; see, e.g., \cite{Schneider:13} and \cite{Cecil:15}.

Fix $\epsilon>0$ small enough. 
Let $\eta(\bz,t)= \phi(\bz + t n(\bz)) J_t(\bz)$ for $(\bz,t)\in\partial\set_u\times[0,\infty)$. Then, we can choose a compact set $K\subset\partial\set_u$ and a constant $t_K>0$ such that for all $t\in[0,t_K]$,
$|\int_{K} \eta(\bz,t) \rd\cH(\bz)-\int_{K} \eta(\bz,0) \rd\cH(\bz)|<{\epsilon}/{3}$ and $\int_{\partial\set_u\setminus K} \eta(\bz,t) \rd\cH(\bz)<{\epsilon}/{3}$.
This is feasible due to the uniform continuity of $\eta$ on $K\times[0,t_K]$, the nonnegativity of $\eta$ on $\partial\set_u\times[0,\infty)$, and the uniform boundedness of $g$ by \cite{Ball:93}. Hence, for all $t\in[0,t_k]$,
\begin{equation}
\begin{aligned}
	&|g(u+t)-g(u)| \\
	&\leq \left|\int_{K} \eta(\bz,t) \rd\cH(\bz)-\int_{K} \eta(\bz,0) \rd\cH(\bz)\right|+\int_{\partial\set_u\setminus K} \eta(\bz,t) \rd\cH(\bz)+\int_{\partial\set_u\setminus K} \eta(\bz,0) \rd\cH(\bz)\\
	&<\epsilon.
\end{aligned}
\end{equation}
Consequently, $p'({u_*}) = g({u_*})>0$.
By the definition of $h$, for any $\varepsilon>0$ small enough, we have
	$u_*^2(p(
	{u_*}+\varepsilon)-p(
	{u_*}))\leq h(
	{u_*}+\varepsilon)-h({u_*})\leq(
	{u_*}+\varepsilon)^2(p({u_*}+\varepsilon)-p({u_*}))$.
Dividing all expressions by $\varepsilon$ and sending $\varepsilon\to0$ result in $h'({u_*})=u_*^2p'({u_*})>0$.
\end{proof}

\bibliographystyle{abbrvnat}
\bibliography{references} 
\end{document}